\documentclass[preprint,12pt]{elsarticle}
\makeatletter
\def\ps@pprintTitle{%
	\let\@oddhead\@empty
	\let\@evenhead\@empty
	\let\@oddfoot\@empty
	\let\@evenfoot\@oddfoot
}
\makeatother
\usepackage{etoolbox}
\makeatletter
\usepackage{setspace}
\usepackage{bm}
\usepackage{graphics}
\usepackage{amsmath, amsthm, amssymb}
\usepackage{natbib}
\usepackage[colorlinks=true]{hyperref}
\usepackage{lscape}
\usepackage[utf8]{inputenc} 
\usepackage{amsthm}

\usepackage{enumerate}
\usepackage{mathrsfs}
\usepackage{placeins}
\usepackage{graphicx}
\usepackage{amsfonts}
\usepackage{verbatim}
\usepackage{amssymb}
\usepackage{amsthm}
\usepackage{float}
\usepackage{multirow}
\newtheorem{thm}{Theorem}[]

\newtheoremstyle{bolditalic}
  {3pt}   
  {3pt}   
  {\itshape} 
  {}      
  {\bfseries} 
  {.}     
  { }     
  {}      

\theoremstyle{bolditalic}
\newtheorem{proposition}{Proposition} 

\theoremstyle{definition}

\theoremstyle{remark}

\usepackage[mathscr]{euscript}
\usepackage{geometry} 
\usepackage{graphicx,lscape} 
\usepackage{booktabs} 
\usepackage{bigstrut}
\usepackage{array} 
\usepackage{paralist} 
\usepackage{verbatim} 
\usepackage{subfig} 

\biboptions{comma,round,authoryear}

\begin{document}
	
	\begin{frontmatter}

		\title{\textbf{How to detect a biased sample using the Rényi divergence measure$?$}}
        
		\author[stat]{Vaishnavi Pavithradas}
\author[stat]{Rajesh G.\corref{cor1}}
\cortext[cor1]{Corresponding author}
\ead{rajeshgstat@gmail.com}
\affiliation[stat]{organization={Department of Statistics},
            addressline={Cochin University of Science and Technology}, 
            city={Cochin},
         postcode={682022}, 
            state={Kerala},
            country={India}}

		\begin{abstract}
			Weighted distributions arise in situations where observations are selected with unequal probabilities or because of the non-observability of some events. Detecting such sampling bias is essential for ensuring valid statistical inference. In this paper, we propose a statistical test for detecting bias in a sample using the Rényi divergence measure. The proposed test statistic is formulated for both uncensored and type-I censored data and possesses several theoretical properties including asymptotic normality. Critical values are obtained through simulations under the Weibull distribution for a range of sample sizes and censoring proportions. A comprehensive power study compares the proposed test with the existing Kullback–Leibler divergence-based test and the likelihood ratio test. Two real datasets are analyzed to demonstrate the practical utility of the proposed test in detecting length bias.
			
		\end{abstract}
		
	\end{frontmatter}
	
	\section{Introduction}
	\noindent
	Size-biased distributions represent a special case of the more general form known as weighted distributions. Originally introduced by \cite{fisher1934effect} to model for ascertainment bias, these distributions were later given a unified theoretical framework by \cite{rao1965discrete}. Weighted distributions commonly arise in practical situations where observations are selected with unequal probabilities, such as in probability proportional to size (PPS) sampling schemes, or because of the non-observability of some events. Briefly, let $X$ be a non-negative random variable with probability density function (pdf) $f(x;\theta)$, with unknown parameter $\theta$, then the corresponding weighted distribution is of the form,
    \begin{equation}\label{wpdf}
        f^w(x;\theta)= \frac{w(x)f(x;\theta)}{E[w(X)]}
    \end{equation}
    where $w(x)$ is a positive weight function and $E[w(X)]$ is the normalizing factor.\\
    A special case occurs when the weight function is expressed as $w(x)=x^r$, distributions of this type are referred to as size-biased distributions of order $r$, then the corresponding pdf is written as (\cite{patil1976size}, \cite{patil1984studies}):
    \begin{equation}\label{lbpdf}
        f_r(x;\theta)= \frac{x^rf(x;\theta)}{E[X^r]}
    \end{equation}
    The most frequently encountered size-biased distributions arise when $r=1$ and $r=2$ in sampling contexts; these cases are commonly referred to as length-biased and area-biased distributions, respectively.  For a review
of weighted distributions and their applications, see, for example, \cite{patil2002weighted}.\\
In practice, detecting biased samples is of considerable importance, since ignoring or failing to identify the presence of a biasing mechanism can lead to conclusions that differ substantially from what would be expected. Being unaware of the sampling mechanism typically results in inconsistent estimators. Consequently, the development of statistical tests to detect size bias is of major importance, and it is this problem that motivates the present study. Existing approaches to this problem include \cite{navarro2003detect} developed parametric
and nonparametric tests for several distributions making use of characterization
properties annd \cite{akman2007simple} provided a test for length-biased
sampling. Additionally, \cite{economou2014kullback}  provideed a Kullback–Leibler divergence–based tests concerning the biasness in  a sample.

Although the Kullback--Leibler divergence has proved useful for this purpose, it represents only a limiting case of the more general Rényi divergence. Introduced by \cite{renyi1961measures}, Rényi divergence forms a one-parameter family of divergence measures that includes the Kullback--Leibler divergence as the limiting case when the order parameter approaches one. The additional order parameter $\alpha$ provides greater flexibility by allowing the divergence to adjust its sensitivity to differences between distributions, thereby offering the potential to improve statistical performance under different sampling schemes. Despite its widespread applications in information theory, statistics, machine learning, and signal processing, the use of Rényi divergence for detecting biased samples has received little attention.\\
Motivated by these, the present paper extends the divergence-based testing framework of \cite{economou2014kullback} by replacing the Kullback--Leibler divergence with the more general Rényi divergence. An empirical Rényi divergence measure is developed for weighted distributions, leading to a new test statistic applicable to both complete and Type-I censored samples. Several theoretical properties of the proposed statistic, including scale invariance, monotonicity with respect to the biasing parameter, and asymptotic normality, are established. Extensive simulation studies are carried out to obtain critical values and evaluate the finite-sample performance of the proposed procedure. The proposed test is further compared with the Kullback--Leibler divergence-based test and the likelihood ratio test, and its practical applicability is demonstrated through two real data examples.

\section{The Rényi measure of divergence for weighted distributions}
In his seminal 1961 work \cite{renyi1961measures}, Rényi introduced a generalized framework for information and divergence measures that serve as a natural extension of Shannon entropy \cite{shannon1948mathematical} and the Kullback–Leibler divergence \cite{kullback1951information}. For a probability density function $f$ and $g$ on $R$, the Rényi divergence of order $\alpha$ between $f$ and $g$ is given by,
\begin{equation}
    D_{\alpha}(f \,\|\, g)
= \frac{1}{\alpha-1} 
\ln  \int f(x)^{\alpha} \, g(x)^{1-\alpha} \, dx \
\end{equation}
where $\alpha> 0$ and $\alpha \neq 0 $. Under appropriate conditions, as $\alpha \to 1$, $D_{\alpha}(f \,\|\, g)$ approaches the Kullback–Leibler divergence between $f$ and $g$. For \( f \) and \( f^w \) as defined in \eqref{wpdf}, the Rényi divergence measure is given by,
\begin{align}
    D_{\alpha}(f \,\|\, f^w)
   & = \frac{1}{\alpha-1} \ln \int_{0}^{\infty} 
       f(x)^{\alpha} \, f^w(x)^{1-\alpha} \, dx  \\
    &= \frac{1}{\alpha-1} \ln \int_{0}^{\infty} 
       f(x)^{\alpha} \, \left(\frac{w(x)f(x)}{E[w(X)]}\right)^{1-\alpha} \, dx \nonumber \\
    &= \frac{1}{\alpha-1} \ln E[w^{1-\alpha}(X)] 
       \, + \, \ln E[w(X)]  \label{wrd}
\end{align}
As mentioned earlier, the weight function for size-biased sampling is given by $w(x)=x^r$ for $r>0$. In this case, the Rényi divergence measure of order $\alpha$ given in \eqref{wrd} is given by,
\begin{equation}\label{lbrd}
  D_{\alpha}(f \,\|\, f_r)  = \frac{1}{\alpha-1} \ln E[X^{r(1-\alpha)}] 
       \, + \, \ln E[X^r]
\end{equation}
These divergence measures can also be extended to situations involving censored observations. Consider type I censoring, where the outcome of the random variable is $X^*=\min(X,x_0)$, with $X$ denoting the underlying study variable and $x_0$ being a predetermined censoring point. In this setting, the corresponding Rényi divergence measure is denoted by $D_{\alpha,c}$ to highlight the presence of censoring and it takes the form,
\begin{equation}\label{censeq}
    D_{\alpha,c}(f \,\|\, f_r)  = \frac{1}{\alpha-1} \ln E[X^{*r(1-\alpha)}] 
       \, + \, \ln E[X^{*r}]
\end{equation}
The subsequent propositions outline several properties of the Rényi divergence measures in the present context.
\begin{proposition}
    The Renyi divergence measure $D_{\alpha,c}(f \,\|\, f_r)$ is an increasing function of $r$.
\end{proposition}
\begin{proof}
Differentiating $D_{\alpha,c}(f \,\|\, f_r)$ with respect to $r$ we get,
\begin{align}
    \frac{d}{dr}D_{\alpha,c}(f \,\|\, f_r)
=-\frac{E[X^{*r(1-\alpha)}\ln X]}{E[X^{*r(1-\alpha)}]}+ \frac{E[X^{*r}\ln X]}{E[X^{*r}]} \label{1}
\end{align}
Since both the functions $\ln x$ and $x^r$ in \eqref{1} are nondecreasing, it follows that $D_{\alpha,c}(f \,\|\, f_r)$ is an increasing function for $r$ when  $\alpha> 0$ and $\alpha \neq 0 $.
\end{proof}
  The following two propositions are useful for simulation studies. Before proceeding, we recall that $\beta$ is a scale parameter if the cdf is $F(x,\beta)$.

\begin{proposition}
The divergence measure $D_{\alpha,c}(f \,\|\, f_r)$ is independent of the scale parameter of the distribution function $F$.
\end{proposition}
\begin{proof}
    The proof follows immediately upon making the substitution $x^*=\beta y^*$ in equation \eqref{censeq}.
\end{proof}

\begin{proposition}\label{prop3}
    If $X$ has a distribution function of the form,\\
    \begin{equation*}
F(x,\beta,\eta)=F\left[\left(\frac{x}{\beta}\right)^\eta,1,1\right]
\end{equation*}
where $\beta$ denotes the scale parameter and $\eta$ the shape parameter, then the Rényi divergence measure $D_{\alpha,c}(f \,\|\, f_r)$ depends only on the ratio $\eta/r$ and the censoring rate $q$.
\end{proposition}
\begin{proof}
Let us consider only the case when $\beta=1$. Accordingly, we write $F(x)$ instead of $F[(x)^\eta,1,1]$. The value $x_q$ represents the censoring point corresponding to the censoring rate $q$, then
\begin{equation*}
    F[(x_q)^{\eta}]= 1-q
\end{equation*}
We have,
\begin{equation*}
    E(X^{*r})= \int_{0}^{x_q} x^r dF(x^\eta) + x_q^r(1-F(x_q^\eta))
\end{equation*}
Through the substitution $u=x^\eta$, the expression can be rewritten as follows:
\begin{align*}
    E(X^{*r})&= \int_{0}^{F^{-1}(1-q)}u^{r/\eta} \,dF(u) + F^{-r/\eta}(1-q)(1-(1-q)) \\
    &=\int_{0}^{F^{-1}(1-q)}u^{r/\eta} \,dF(u) + qF^{-r/\eta}(1-q)
\end{align*}
which depends only on the ratio $\eta/r$ and the censoring rate $q$. Likewise we obtain,
\begin{equation*}
    E(X^{*r(1-\alpha)})=  \int_{0}^{x_q} x^{r(1-\alpha)} dF(x^\eta) + x_q^{r(1-\alpha)}(1-F(x_q^\eta))
\end{equation*}
Using the transformation as above, we get,
\begin{equation*}
    E(X^{*r(1-\alpha)})= \int_{0}^{F^{-1}(1-q)}u^{\frac{r}{\eta}(1-\alpha)} \,dF(u) + qF^{-\frac{r}{\eta}(1-\alpha)}(1-q)
\end{equation*}
which concludes the proof.
\end{proof}
\section{The empirical measure of divergence and the corresponding proposed test}
In this section, we develop tests for detecting bias in a sample using the Rényi divergence between the density functions $f$ and $ f_w$. We also define the empirical measure of $D_{\alpha,c}(f \,\|\, f^w)$ and present the proposed tests. While earlier studies have employed divergence measures—such as Kullback-Leibler divergence for detecting bias \cite{economou2014kullback} and Rényi distance of the equilibrium distributions is constructed to test exponentiality \cite{sadeghpour2018exponentiality}. Here we want to test the hypothesis
\begin{equation}\label{hypothesis1}
    H_0 : w(x) = 1~~~vs.~~~ H_1: w(x) = w_0(x)
\end{equation}
Although the proposed methods are formulated for a general weight function, particular emphasis is placed on the size-biased context, as it is one of the most commonly encountered cases in practice. In this context, the preceding hypotheses can be reformulated as follows.
\begin{equation}\label{hypothesis2}
    H_0 : r = 0~~~vs.~~~ H_1: r = r_0
\end{equation}
where $w_0(x)$ or $r_0$ is fixed and $r_0$ takes the values $1$ or $2$.
\subsection{The empirical measure of Rényi divergence}
Let $x_1^*,x_2^*,\ldots, x_n^*$ be an observed sample of size $n$ from the true distribution with density function $f(\cdot)$, then the expected values $E[w^{1-\alpha}(X^*)]$ and $E[w(X^*)]$ can be approximated by,
\begin{equation*}
    E[w^{1-\alpha}(X^*)] \approx \frac{1}{n}\sum_{i=1}^{n}w^{1-\alpha}(x^*)
\end{equation*}
also 
\begin{equation*}
    E[w(X^*)] \approx \frac{1}{n}\sum_{i=1}^{n}w(x^*)
\end{equation*}
Thus, the expression for $D_{\alpha,c}(f \,\|\, f^w)$ in equation \eqref{wrd} may be approximated as follows,
\begin{align}
    D_{\alpha,c}^E(f \,\|\, f^w)&= \frac{1}{\alpha-1}\, \ln \left[\frac{1}{n}\sum_{i=1}^{n}w^{1-\alpha}(x^*)\left[\frac{1}{n}\sum_{i=1}^{n}w(x^*)\right]^{\alpha-1}\right] \\ \nonumber
    &= \frac{1}{\alpha-1}\ln \left[\frac{\frac{1}{n}\sum_{i=1}^{n}w^{1-\alpha}(x^*)}{[\frac{1}{n}\sum_{i=1}^{n}w(x^*)]^{1-\alpha}}\right]
    \end{align}
    which can be expressed as,
    \begin{align}
  D_{\alpha,c}^E(f \,\|\, f^w)  =-\, \ln\left(\frac{\text{Power mean of order $(1-\alpha)$ of $w(x_i^*)$ }}{\text{Arithmetic mean of $w(x_i^*)$}}\right) \label{statistic}
\end{align}
When $\alpha \to 1$, that is, in the case of Kullback-Leibler divergence, the numerator, which is the power mean of order $(1-\alpha)$, becomes the geometric mean of $w(x_i^*)$. The above statistics $D_{\alpha,c}^E(f \,\|\, f^w)$ can be used for testing the hypothesis \eqref{hypothesis1}. Therefore, the null hypothesis is rejected at the significance level $\alpha_c$ whenever,
\begin{equation*}
    D_{\alpha,c}^E(f \,\|\, f^w) < u(\alpha_c,q,n)
\end{equation*}
where $u(\alpha_c,q,n)$ denotes the corresponding critical value.
Equivalently, the test rejects $H_0$ when,
\begin{equation*}
    \gamma > u(\alpha_c,q,n)
\end{equation*}
where,
\begin{equation}\label{teststat}
    \gamma= \frac{\text{Power mean of order $(1-\alpha)$ of $w_0(x_i^*)$ }}{\text{Arithmetic mean of $w_0(x_i^*)$}}
\end{equation}
In the context of size-bias sampling, the test statistics $\gamma$ given in \eqref{teststat}, can be expressed as follows,
\begin{equation}\label{lengthteststat}
    \gamma= \frac{\left[\frac{1}{n}\sum_{i=1}^{n}x^{*r(1-\alpha)}\right]^{\frac{1}{1-\alpha}}}{\bar{x^{*r}}}
\end{equation}
This statistic retains some properties of the Rényi divergence. In particular,  we can prove that it is independent of the scale parameter. \\
The large sample properties of the following statistic \eqref{lengthteststat} can be derived using the following theorem.
\begin{thm}
    Let $x_1^*,x_2^*, \ldots,x_n^*$ be a censored sample from $X$ with pdf $f(\cdot)$. Then the asymptotic distribution of $\gamma$ as defined in \eqref{teststat}, under the null hypothesis \eqref{hypothesis1} is normal with,
    \begin{equation*}
        \gamma \xrightarrow{D} N\left[\frac{\mu^{1/1-\alpha}}{E(w(X^*))},\frac{\sigma^2\mu^{\frac{2\alpha}{1-\alpha}}}{(1-\alpha)^2[E(w(X^*))]^2}\right]
  \end{equation*}
    \end{thm}
\begin{proof}
We have,
\begin{equation*}
\gamma= \frac{\left[\frac{1}{n}\sum_{i=1}^{n}w^{1-\alpha}(x_i^*)\right]^{\frac{1}{1-\alpha}}}{\frac{1}{n}\sum_{i=1}^{n}w(x_i^*)}
\end{equation*}
Let $\mu= E[w^{1-\alpha}(x_i^*)]$ and $\sigma^2 = V[w^{1-\alpha}(x_i^*)]$\\
By the Central Limit Theorem we have,
\begin{equation*}
    \sqrt{n}\left[\frac{1}{n}\sum_{i=1}^{n}w^{1-\alpha}(x_i^*)-\mu\right] \xrightarrow{D} N(0,\sigma^2)
\end{equation*}
Now consider the function $g(x)= x^{1/1-\alpha}$\\
The Delta method applied to $g\left(\frac{1}{n}\sum_{i=1}^{n}w^{1-\alpha}(x_i^*)\right)$ yields,
\begin{equation}\label{a}
     \sqrt{n}\left[\left[\frac{1}{n}\sum_{i=1}^{n}w^{1-\alpha}(x_i^*)\right]^{1/1-\alpha}-{\mu}^{1/1-\alpha}\right] \xrightarrow{D} N\left[0,\sigma^2\left[\frac{\mu^{\alpha/1-\alpha}}{1-\alpha}\right]^2\right]
\end{equation}
On the other hand,
\begin{equation}\label{b}
    \frac{1}{n}\sum_{i=1}^{n}w(x_i^*) \xrightarrow{P} E[w(x_i^*)]
\end{equation}
Finally, with the use of Slutsky's theorem, we obtain from \eqref{a} and \eqref{b},
\begin{equation}
    \sqrt{n}\left[\frac{\left[\frac{1}{n}\sum_{i=1}^{n}w^{1-\alpha}(x_i^*)\right]^{1/1-\alpha}-{\mu}^{1/1-\alpha}}{ \frac{1}{n}\sum_{i=1}^{n}w(x_i^*)}\right] \xrightarrow{D} N\left[0,\frac{\sigma^2\left[\frac{\mu^{\alpha/1-\alpha}}{1-\alpha}\right]^2}{ [\frac{1}{n}\sum_{i=1}^{n}w(x_i^*)]^2}\right] \\
   \end{equation}
   Thus,
   \begin{equation}
    \sqrt{n}\left[\gamma-\frac{{\mu}^{1/1-\alpha}}{E[w(x_i^*)]}\right] \xrightarrow{D} N\left[0,\frac{\sigma^2\mu^{\frac{2\alpha}{1-\alpha}}}{(1-\alpha)^2[E(w(x_i^*))]^2}\right]
\end{equation}
which completes the proof.
\end{proof}
    
\section{Numerical analysis}

In this section, we present the critical values of the proposed test statistic $\gamma$, based on the Rényi divergence measure, using the weight function $w(x)=x^{r_0}$, for the Weibull distribution $W(\beta,\eta)$. The Weibull distribution is adopted in this study since it satisfies the conditions stated in Proposition \ref{prop3}, making it suitable for the proposed test. The critical values are obtained for various sample sizes, censoring rates, and choices of the shape parameter. Additionally, a comprehensive power study is conducted to assess the performance and effectiveness of the proposed test.

\subsection{Simulation study}
Due to the simple analytical structure of the empirical Rényi divergence measure, the computation of critical values can be carried out efficiently using R software. At the significance level
$\alpha_c=0.05$, critical values for the proposed Rényi divergence–based test with weight function $w(x)=x^{r_0}$ are obtained through simulation for various sample sizes, censoring proportions, and values of the shape parameter $\eta$. Since the test statistic $\gamma$ is invariant with respect to the scale parameter $\beta$ and depends only on the ratio $\eta/r$, the simulation study is conducted by fixing $\beta=1$.
\par
Overall, the proposed test exhibits satisfactory performance across a wide range of settings. However, its performance may be affected by the presence of extreme observations, such as very large values or values close to zero. This behavior arises from the sensitivity of the test statistic $\gamma$  because of the denominator. The $\gamma$ is sensitive to large values. Consequently, extreme values of $\gamma$ are interpreted as evidence in favor of the alternative hypothesis, and appropriate lower and upper threshold values are introduced to define the rejection region. Thus, the null hypothesis is rejected when
\begin{equation*}
    \gamma < l(\alpha_{2},q,n)~~~or~~~ \gamma > u(\alpha_1,q,n) 
\end{equation*}
\par
The critical values are selected to preserve the overall significance level $\alpha_c$ of the test. Following the significance-level partitioning approach of \cite{economou2014kullback}, the level $\alpha_c$ is divided as $\alpha_1=p\alpha_c$ and $\alpha_2=(1-p)\alpha_c$, where $0\leq p\leq 1$. They examined several values of $p$, and the choice of $p=0.90$ was found to provide an appropriate. Accordingly, throughout the subsequent analysis, we adopt $\alpha_1=0.9\alpha_c$ and $\alpha_2=0.1\alpha_c$, and, since the Rényi divergence test statistic $\gamma$ is related to $\alpha$, its power performance is examined for different values of $\alpha$. \noindent \par
From Figure \ref{fig1}, it is observed that as the sample size $n$ increases, the average power increases. Moreover, the maximum happens when $\alpha$ is $0.5$. Consequently, $\alpha=0.5$ is selected for the computation of critical values and power comparisons in the subsequent analysis.
\begin{figure}[h]
\centering

\includegraphics[width=0.9\textwidth]{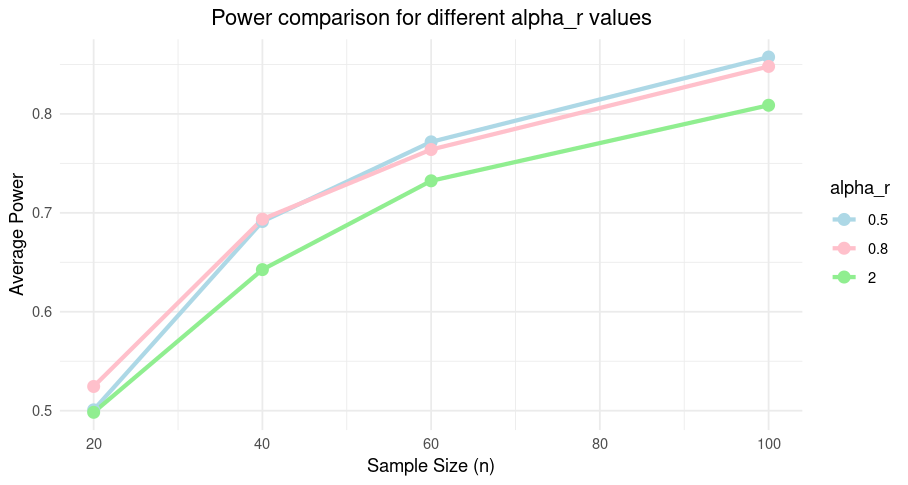}
\caption{}\label{fig1}
\end{figure}

\subsubsection{Critical values of $\gamma$ for the uncensored case}

The 10,000 random samples were generated from the Weibull distribution in order to find the critical values of $\gamma$, that is, to find the critical values for testing the hypothesis \eqref{hypothesis2}. Samples were generated from $W(\beta=1,\eta)$ distribution for different sample sizes $n$ and shape parameter $\eta$. The resulting lower and critical values corresponding to $\alpha_c=0.05,~p=0.9,~ \text{and}~ q=0$ are presented in Tables \ref{tab:lowercritical_gamma_uncensored} and \ref{tab:uppercritical_gamma_uncensored}, respectively. 
An inspection of these tables reveals that the lower critical values of $\gamma$ increase with an increase in either the sample size or the shape parameter, whereas the upper critical values decrease as the sample size increases and increase with larger values of $\eta$.

\begin{table}[H]
\centering
\caption{The lower critical values $l(0.005,0,n,\eta)$ of $\gamma$ for the Weibull distribution and $r_0=1$. }
\label{tab:lowercritical_gamma_uncensored}
\small
\begin{tabular}{c cccccccccc}
\toprule
  \multicolumn{10}{c}{Sample size $(n)$} \\
\cmidrule(lr){2-11}
$\eta$ 
& 20 & 40 & 60 & 80 & 100 & 150 & 200 & 350 & 500 & 1000 \\
\midrule
0.3 & 0.1196 & 0.1043 & 0.1064 & 0.1146 & 0.1146 & 0.1259 & 0.1324 & 0.1479 & 0.1586 & 0.1806 \\
0.4 & 0.1908 & 0.2134 & 0.2225 & 0.2393 & 0.2401 & 0.2651 & 0.2752 & 0.3046 & 0.3163 & 0.3360 \\
0.5 & 0.2916 & 0.3248 & 0.3535 & 0.3677 & 0.3777 & 0.3931 & 0.4100 & 0.4296 & 0.4429 & 0.4594 \\
0.6 & 0.3869 & 0.4305 & 0.4566 & 0.4693 & 0.4804 & 0.5019 & 0.5120 & 0.5310 & 0.5406 & 0.5532 \\
0.7 & 0.4701 & 0.5208 & 0.5445 & 0.5570 & 0.5629 & 0.5816 & 0.5879 & 0.6072 & 0.6155 & 0.6279 \\
0.8 & 0.5426 & 0.5918 & 0.6074 & 0.6194 & 0.6313 & 0.6443 & 0.6537 & 0.6664 & 0.6740 & 0.6834 \\
0.9 & 0.5968 & 0.6458 & 0.6586 & 0.6731 & 0.6827 & 0.6968 & 0.7018 & 0.7143 & 0.7203 & 0.7290 \\
1.0 & 0.6473 & 0.6904 & 0.7034 & 0.7181 & 0.7229 & 0.7347 & 0.7423 & 0.7515 & 0.7575 & 0.7656 \\
1.1 & 0.6896 & 0.7240 & 0.7396 & 0.7525 & 0.7568 & 0.7653 & 0.7730 & 0.7831 & 0.7879 & 0.7951 \\
1.2 & 0.7215 & 0.7541 & 0.7692 & 0.7787 & 0.7829 & 0.7946 & 0.7993 & 0.8087 & 0.8131 & 0.8200 \\
1.3 & 0.7504 & 0.7796 & 0.7937 & 0.8026 & 0.8081 & 0.8164 & 0.8214 & 0.8312 & 0.8335 & 0.8398 \\
1.4 & 0.7729 & 0.8017 & 0.8139 & 0.8209 & 0.8271 & 0.8364 & 0.8399 & 0.8472 & 0.8516 & 0.8566 \\
1.5 & 0.7905 & 0.8200 & 0.8324 & 0.8390 & 0.8444 & 0.8524 & 0.8556 & 0.8621 & 0.8665 & 0.8711 \\
2.0 & 0.8638 & 0.8838 & 0.8913 & 0.8985 & 0.8992 & 0.9054 & 0.9078 & 0.9136 & 0.9153 & 0.9188 \\
3.0 & 0.9255 & 0.9383 & 0.9440 & 0.9466 & 0.9488 & 0.9514 & 0.9532 & 0.9558 & 0.9572 & 0.9592 \\
4.0 & 0.9536 & 0.9615 & 0.9655 & 0.9676 & 0.9685 & 0.9706 & 0.9713 & 0.9735 & 0.9742 & 0.9754 \\
8.0 & 0.9861 & 0.9887 & 0.9899 & 0.9906 & 0.9909 & 0.9915 & 0.9920 & 0.9926 & 0.9928 & 0.9932 \\
\bottomrule
\end{tabular}
\end{table}

\begin{table}[H]
\centering
\caption{The upper critical values $u(0.045,0,n,\eta)$ of $\gamma$ for the Weibull distribution and $r_0=1$. }
\label{tab:uppercritical_gamma_uncensored}
\small
\begin{tabular}{c cccccccccc}
\toprule
  \multicolumn{10}{c}{Sample size $(n)$} \\
\cmidrule(lr){2-11}
$\eta$ 
& 20 & 40 & 60 & 80 & 100 & 150 & 200 & 350 & 500 & 1000 \\
\midrule
0.3 & 0.5128 & 0.4361 & 0.4003 & 0.3817 & 0.3681 & 0.3459 & 0.3320  & 0.3115 & 0.3029 & 0.2862 \\
0.4 & 0.6128 & 0.5465 & 0.5167 & 0.5009 & 0.4906 & 0.4712 & 0.4581 & 0.4424 & 0.4329 & 0.4194 \\
0.5 & 0.6906 & 0.6362 & 0.6119 & 0.5958 & 0.5851 & 0.5692 & 0.5605 & 0.5460  & 0.5380  & 0.5271 \\
0.6 & 0.7518 & 0.7022 & 0.6825 & 0.6685 & 0.6611 & 0.646  & 0.6382 & 0.6265 & 0.6195 & 0.6109 \\
0.7 & 0.7970  & 0.7556 & 0.7368 & 0.7249 & 0.7179 & 0.7058 & 0.6990  & 0.6889 & 0.6829 & 0.6754 \\
0.8 & 0.8301 & 0.7943 & 0.7790  & 0.7696 & 0.7627 & 0.7520  & 0.7471 & 0.7384 & 0.7332 & 0.7257 \\
0.9 & 0.8569 & 0.8265 & 0.8126 & 0.8034 & 0.7981 & 0.7902 & 0.7850  & 0.7771 & 0.7724 & 0.7658 \\
1.0   & 0.8775 & 0.8504 & 0.8397 & 0.8314 & 0.8260  & 0.8191 & 0.8147 & 0.8074 & 0.8035 & 0.7981 \\
1.1 & 0.8954 & 0.8713 & 0.8605 & 0.8539 & 0.8500   & 0.8428 & 0.8386 & 0.8326 & 0.8290  & 0.8245 \\
1.2 & 0.9090  & 0.8877 & 0.8785 & 0.8734 & 0.8692 & 0.8626 & 0.8588 & 0.8529 & 0.8497 & 0.8458 \\
1.3 & 0.9201 & 0.9014 & 0.8930  & 0.8868 & 0.8838 & 0.8788 & 0.8750  & 0.8700   & 0.8675 & 0.8638 \\
1.4 & 0.9300   & 0.9132 & 0.9054 & 0.9006 & 0.8966 & 0.8921 & 0.8890  & 0.8843 & 0.8819 & 0.8786 \\
1.5 & 0.9383 & 0.9215 & 0.9147 & 0.9110  & 0.9086 & 0.9036 & 0.9006 & 0.8967 & 0.8945 & 0.8913 \\
2.0   & 0.9627 & 0.9526 & 0.9486 & 0.9457 & 0.9439 & 0.9407 & 0.9387 & 0.9359 & 0.9345 & 0.9324 \\
3.0   & 0.9829 & 0.9776 & 0.9752 & 0.9739 & 0.9727 & 0.9713 & 0.9702 & 0.9687 & 0.9679 & 0.9667 \\
4.0   & 0.9899 & 0.9871 & 0.9856 & 0.9847 & 0.9841 & 0.9831 & 0.9825 & 0.9816 & 0.9810  & 0.9803 \\
8.0   & 0.9974 & 0.9966 & 0.9962 & 0.9960  & 0.9958 & 0.9956 & 0.9953 & 0.9951 & 0.9949 & 0.9947\\
\bottomrule
\end{tabular}
\end{table}

\subsubsection{Critical values of $\gamma$ for the censored case}

For the censored sample setting, the lower and upper critical values, denoted by $\ell(0.005, q, n, \eta)$ and $u(0.045, q,n,\eta)$, were obtained for censoring proportions of $q=20\%$ and $q=40\%$. Similar to the uncensored case, $10,000$ random samples were generated for each combination of sample sizes $n = 20, 40, 60, 100, 350, 500, \text{and}~ 1000$, considering the same set of values for the Weibull shape parameter $\eta$. The censoring threshold $x_0$ was selected in such a way that approximately $q\%$ of the observations were censored, satisfying the conditions $F(x_0)= 1-q$. The corresponding critical values are reported in Tables \ref{table 3} and \ref{table 4}. The numerical results indicate that the monotonic behaviour observed in the uncensored case remains valid under censoring as well, with both lower and upper critical values showing systematic variation with respect to the sample size $n$ and the parameter $\eta$. For fixed  $n$ and $\eta$, the lower and upper critical values are consistently large for $q=0.4$ than for $q=0.2$.
It is also observed that these critical values become larger as the censoring proportion increases. Furthermore, the rate of convergence towards the asymptotic distribution is faster for larger values of $\eta$.


\begin{landscape}

\begin{table}[htbp]
\centering
\caption{The lower $\ell(0.005, q, n, \eta)$ critical values of the $\gamma$ test for the Weibull distribution for $r_0 = 1$ and $q = 0.2, 0.4$.}

\renewcommand{\arraystretch}{1.15}
\small   

\begin{tabular}{c ccccccc ccccccc}
\toprule

\multirow{3}{*}{$\eta$}
& \multicolumn{7}{c}{$q = 0.2$}
& \multicolumn{7}{c}{$q = 0.4$} \\

\cmidrule(lr){2-8}
\cmidrule(lr){9-15}

& \multicolumn{7}{c}{Sample size}
& \multicolumn{7}{c}{Sample size} \\

\cmidrule(lr){2-8}
\cmidrule(lr){9-15}

& 20 & 40 & 60 & 100 & 350 & 500 & 1000
& 20 & 40 & 60 & 100 & 350 & 500 & 1000 \\

\midrule

0.3
& 0.3312&0.3788&0.4066&0.431&0.472&0.4824&0.4932
& 0.4542&0.5185&0.5478&0.5742&0.6183&0.6283&0.6403 \\

0.4
& 0.4141&0.475&0.4931&0.519&0.5594&0.5659&0.5767
& 0.5347&0.5888&0.6162&0.6422&0.6806&0.6879&0.6994 \\

0.5
& 0.4992&0.5485&0.5668&0.5888&0.6286&0.6345&0.6447
& 0.5963&0.6446&0.6734&0.6925&0.7309&0.7374&0.7477 \\

0.6
& 0.5596&0.6067&0.6304&0.6469&0.6838&0.6882&0.6982
& 0.6487&0.6909&0.7172&0.7354&0.7699&0.7752&0.7833\\

0.7
& 0.6139&0.6561&0.6771&0.6964&0.7272&0.7331&0.7417
& 0.6924&0.7336&0.7512&0.7696&0.7998&0.8058&0.8143\\

0.8
& 0.6592&0.7022&0.7186&0.7343&0.7637&0.7695&0.7757
& 0.7252&0.7667&0.7819&0.7978&0.8257&0.8309&0.8384\\

0.9
& 0.6966&0.736&0.7518&0.7653&0.7933&0.7978&0.8051
& 0.754&0.7888&0.8048&0.8216&0.8477&0.8518&0.8584\\

1.0
& 0.7287&0.7643&0.7786&0.7958&0.8176&0.822&0.829
& 0.7819&0.8157&0.8289&0.8415&0.8648&0.8681&0.8746\\

1.1
& 0.758&0.7864&0.8034&0.8152&0.8389&0.8427&0.8481
& 0.8027&0.8311&0.8479&0.8588&0.8792&0.883&0.8889\\

1.2
& 0.7796&0.8109&0.8208&0.8358&0.8557&0.859&0.865
& 0.8178&0.8491&0.8594&0.8737&0.8923&0.8956&0.9004\\

1.3
& 0.7981&0.8275&0.8395&0.8504&0.8692&0.8735&0.8785
& 0.8349&0.861&0.872&0.8848&0.9019&0.905&0.9101 \\

1.4
&0.8172&0.8444&0.8547&0.8645&0.8829&0.8861&0.8908

&
0.8473&0.8725&0.8846&0.8945&0.9115&0.9145&0.9187
\\
1.5
&
0.8317&0.857&0.8678&0.8777&0.8937&0.8966&0.9011

&
0.8596&0.8834&0.8908&0.9046&0.919&0.9222&0.9261
\\
2.0
&
0.8863&0.9044&0.9105&0.9188&0.9302&0.9327&0.9358

&
0.9019&0.9211&0.9275&0.9356&0.9465&0.9486&0.9514
\\
3.0
&
0.9363&0.9473&0.9526&0.9571&0.9639&0.965&0.9667

&
0.945&0.9558&0.9607&0.965&0.9716&0.9728&0.9744
\\
4.0
&
0.9601&0.9681&0.9698&0.9733&0.9779&0.9785&0.9797
&
0.9647&0.9725&0.9754&0.9777&0.9824&0.9832&0.9842
\\

8.0
&
0.9869&0.9899&0.9911&0.9922&0.9936&0.9939&0.9943
&
0.9887&0.9915&0.9924&0.9934&0.9949&0.9951&0.9955
\\
\bottomrule
\end{tabular}
\label{table 3}
\end{table}
\end{landscape}

\begin{landscape}

\begin{table}[htbp]
\centering
\caption{The upper $u(0.045, q, n, \eta)$ critical values of the $\gamma$ test for the Weibull distribution for $r_0 = 1$ and $q = 0.2, 0.4$.}

\renewcommand{\arraystretch}{1.15}
\small   

\begin{tabular}{c ccccccc ccccccc}
\toprule

\multirow{3}{*}{$\eta$}
& \multicolumn{7}{c}{$q = 0.2$}
& \multicolumn{7}{c}{$q = 0.4$} \\

\cmidrule(lr){2-8}
\cmidrule(lr){9-15}

& \multicolumn{7}{c}{Sample size}
& \multicolumn{7}{c}{Sample size} \\

\cmidrule(lr){2-8}
\cmidrule(lr){9-15}

& 20 & 40 & 60 & 100 & 350 & 500 & 1000
& 20 & 40 & 60 & 100 & 350 & 500 & 1000 \\

\midrule

0.3
& 0.6665&0.6249&0.6068&0.586&0.5564&0.5507&0.5434
&0.8067&0.7689&0.7505&0.7328&0.7035&0.6986&0.6897 \\

0.4
& 0.734&0.6946&0.6784&0.6635&0.6349&0.6295&0.6229
&0.8473&0.8116&0.7994&0.7822&0.7565&0.7517&0.7443\\

0.5
&0.7883&0.7512&0.7376&0.7214&0.6968&0.6926&0.6857
& 0.8805&0.8482&0.8335&0.8202&0.7973&0.7931&0.7864\\

0.6
& 0.8271&0.7961&0.7814&0.7677&0.7462&0.7416&0.7356
&0.8989&0.8739&0.8622&0.8491&0.829&0.8255&0.8195\\

0.7
& 0.8594&0.8275&0.8173&0.8051&0.7848&0.7808&0.7754
& 0.9176&0.8941&0.8833&0.8728&0.8547&0.8514&0.846\\

0.8
& 0.8797&0.8561&0.8452&0.8324&0.8156&0.8118&0.8072
& 0.9316&0.9109&0.9006&0.8913&0.875&0.872&0.8669\\

0.9
& 0.8977&0.8762&0.8672&0.8566&0.8404&0.8375&0.8327
&0.9428&0.9228&0.914&0.9056&0.8909&0.8885&0.8843\\

1.0
&0.9145&0.8931&0.8842&0.876&0.8611&0.8583&0.8541
& 0.9501&0.9333&0.9258&0.9175&0.905&0.9023&0.8983\\

1.1
& 0.924&0.9077&0.8989&0.8912&0.8777&0.875&0.8713
& 0.9571&0.942&0.9346&0.928&0.916&0.9136&0.9103\\

1.2
&0.9354&0.9181&0.9107&0.9039&0.892&0.8895&0.8861
& 0.9627&0.9498&0.9423&0.9358&0.9253&0.9231&0.9201\\

1.3
&0.9429&0.9278&0.9208&0.9148&0.9038&0.9015&0.8983
& 0.9671&0.9553&0.9493&0.9434&0.9332&0.931&0.9282 \\

1.4
&0.9499&0.936&0.9303&0.9237&0.9136&0.9117&0.9087

&
0.9717&0.9599&0.9547&0.9488&0.9397&0.9382&0.9355
\\
1.5
&
0.9556&0.9427&0.9376&0.9313&0.9224&0.9204&0.9178

&
0.9742&0.9641&0.9591&0.9541&0.9456&0.944&0.9416
\\
2.0
&
0.9731&0.965&0.9612&0.9576&0.9507&0.9495&0.9476

&
0.9843&0.9778&0.9745&0.9708&0.9651&0.964&0.9623
\\
3.0
&
0.9872&0.9832&0.9809&0.9789&0.9753&0.9747&0.9735

&
0.9924&0.989&0.9874&0.9853&0.9822&0.9816&0.9807
\\
4.0
&
0.9927&0.9903&0.9888&0.9875&0.9852&0.9848&0.9842
&
0.9957&0.9936&0.9925&0.9913&0.9893&0.989&0.9883
\\

8.0
&
0.9981&0.9974&0.997&0.9966&0.9959&0.9958&0.9956
&
0.9989&0.9983&0.998&0.9976&0.997&0.9969&0.9967
\\
\bottomrule
\end{tabular}
\label{table 4}
\end{table}
\end{landscape}

\subsection{Power study}
To assess the performance of the proposed Rényi divergence-based test statistic, a simulation study was carried out and its power was compared with the Kullback–Leibler (KL) divergence-based test and the Likelihood Ratio Test (LRT) proposed in \cite{economou2014kullback}. Since the KL divergence test and the LRT are well-established procedures for testing biasness in a sample, the comparison provides a useful framework for evaluating the effectiveness of the proposed method.
\begin{enumerate}
    \item  The Kullback–Leibler measure of divergence based test statistic:
    \begin{equation}
        \lambda_1 = \frac{\text{Geometric Mean}~ w_0(x_i^*)}{\text{Arithmetic Mean of}~w_0(x_i^*)}
    \end{equation}
    \item Likelihood Ratio Test:
    \begin{equation}
         \Lambda=\frac{\prod_{i=1}^{n}w(x_i)}{[E(w(X)]^n}
    \end{equation}
    where $
    \hat{\theta}_w$ and $\hat{\theta}$ are the estimation of $\theta$ under the $f_w$ and the $f$ model respectively.
    \end{enumerate}
    The study was conducted under the hypothesis,
    \begin{equation*}
         H_0 : r = 0~~~vs.~~~ H_1: r = r_0
    \end{equation*}
    where $r_0=1,2$ respectively. Random samples were generated from the Weibull distribution with scale parameter $\beta=1$ and shape parameter $\eta= 0.5,1,2, \text{and}~4,$ for sample sizes $n=20,40,60, 
    \text{and}~ 100,$ at significance level $\alpha_c=0.05$.
   \\
   For the proposed Rényi divergence test, 2000 biased samples were generated under the alternative hypothesis and the power was estimated using the simulated lower and upper critical values obtained previously. The same combinations of sample size and parameter values were considered as in the KL divergence framework to ensure direct comparability. For censored samples, censoring proportions of $q=0.2$ and $q=0.4$ were also incorporated.
   \\
    The power values of the Rényi divergence test given in Table \ref{table 5}, \ref{table 6} and \ref{table 7}, were then compared with the corresponding powers of the KL divergence-based test statistic $\lambda_1$ and the LRT reported in the reference study. From the results, we can see that the power $\gamma$ decreases as the ratio $\eta/r$ increases. When $q=0$, in the uncensored case, the LR test performs better. In the censored cases, when compared with Kullback-Leibler divergence test and LR test, the Renyi divergence test presents higher power.

    \begin{table}[ht!]
\centering
\caption{Estimated power of $\gamma$, $\lambda_1$ and LR test for $q=0$}
\renewcommand{\arraystretch}{1.2}
\begin{tabular}{cccccccccc}
\hline
 &  & \multicolumn{3}{c}{\(H_0:r=1\)} &  &  & \multicolumn{3}{c}{\(H_0:r=2\)} \\ 
\cline{3-5} \cline{8-10}

\(n\) & \(\eta\) & \(\gamma\) & \(\lambda_1\) & LR 
& 
\(n\) & \(\eta\) & \(\gamma\) & \(\lambda_1\) & LR \\

\hline

20  & 0.5 &0.8840  & 0.9675 & 0.9565 
& 20  & 1   &0.8885  & 0.9570 & 0.9625 \\

20  & 1.0 &0.6570  & 0.6710 & 0.8905 
& 20  & 2   &0.6605  & 0.6670 & 0.8850 \\

20  & 2.0 &0.2970  & 0.3140 & 0.7725 
& 20  & 4   &0.3240  & 0.3080 & 0.4395 \\

20  & 4.0 &0.1545  & 0.1345 & 0.3150 
& 20  & 8   &0.1390  & 0.1535 & 0.7215 \\

\hline

40  & 0.5 &0.9885  & 1.0000 & 0.9965 
& 40  & 1   &0.9910  & 1.0000 & 0.9975 \\

40  & 1.0 &0.9310  & 0.9560 & 0.9595 
& 40  & 2   &0.9310  & 0.9615 & 0.9610 \\

40  & 2.0 &0.5850  & 0.5975 & 0.8750 
& 40  & 4   & 0.5450 & 0.5965 & 0.6420 \\

40  & 4.0 &0.2475  & 0.2405 & 0.4985 
& 40  & 8   & 0.2515 & 0.2530 & 0.7425 \\

\hline

60  & 0.5 &0.9985  & 1.0000 & 1.0000 
& 60  & 1   & 0.9999 & 1.0000 & 0.9995 \\

60  & 1.0 &0.9910  & 0.9960 & 0.9865 
& 60  & 2   &0.9860  & 0.9945 & 0.9775 \\

60  & 2.0 &0.7440  & 0.7490 & 0.8995 
& 60  & 4   &0.7720  & 0.7510 & 0.7390 \\

60  & 4.0 &0.3405  & 0.3520 & 0.6210 
& 60  & 8   &0.3310  & 0.3455 & 0.7820 \\

\hline

100 & 0.5 &1.0000  & 1.0000 & 1.0000 
& 100 & 1   &1.0000  & 1.0000 & 1.0000 \\

100 & 1.0 &1.0000  & 1.0000 & 0.9960 
& 100 & 2   &1.0000  & 1.0000 & 0.9970 \\

100 & 2.0 &0.9255  & 0.9240 & 0.9490 
& 100 & 4   &0.9340  & 0.9235 & 0.8540 \\

100 & 4.0 &0.4995  & 0.5070 & 0.7350 
& 100 & 8   &0.4610  & 0.4990 & 0.8410 \\

\hline
\end{tabular}
\label{table 5}
\end{table}

\begin{table}[]
\centering
\caption{Estimated power of $\gamma$, $\lambda_1$ and LR test for $q=0.2$}
\renewcommand{\arraystretch}{1.2}
\begin{tabular}{cccccccccc}
\hline
 &  & \multicolumn{3}{c}{\(H_0:r=1\)} &  &  & \multicolumn{3}{c}{\(H_0:r=2\)} \\
\cline{3-5} \cline{8-10}

\(n\) & \(\eta\)& \(\gamma\) & \(\lambda_1\) & LR
&
\(n\) & \(\eta\) & \(\gamma\) & \(\lambda_1\) & LR \\

\hline

20 & 0.5&1.0000 & 0.9845 & 0.9710
& 20 & 1&1.0000 & 0.9865 & 0.9690 \\

20 & 1.0&0.9735 & 0.7210 & 0.6950
& 20 & 2&0.9775 & 0.7195 & 0.6765 \\

20 & 2.0&0.5280 & 0.3150 & 0.3430
& 20 & 4&0.5560 & 0.3165 & 0.3410 \\

20 & 4.0&0.2105 & 0.1590 & 0.1595
& 20 & 8&0.2140 & 0.1555 & 0.1435 \\

\hline

40 & 0.5&1.0000 & 1.0000 & 1.0000
& 40 & 1 &1.0000 & 1.0000 & 1.0000 \\

40 & 1.0&1.0000 & 0.9625 & 0.9375
& 40 & 2&1.0000 & 0.9600 & 0.9595 \\

40 & 2.0&0.8365 & 0.5850 & 0.6245
& 40 & 4&0.8130 & 0.5650 & 0.5945 \\

40 & 4.0&0.3635 & 0.2390 & 0.2765
& 40 & 8&0.3465 & 0.2420 & 0.2645 \\

\hline

60 & 0.5&1.0000 & 1.0000 & 1.0000
& 60 & 1&1.0000 & 1.0000 & 1.0000 \\

60 & 1.0&1.0000 & 0.9955 & 0.9930
& 60 & 2&1.0000 & 0.9955 & 0.9905 \\

60 & 2.0&0.9340 & 0.7375 & 0.7790
& 60 & 4&0.9460 & 0.7400 & 0.7665 \\

60 & 4.0&0.4660 & 0.3490 & 0.4015
& 60 & 8&0.4625 & 0.3380 & 0.3695 \\

\hline

100 & 0.5&1.0000 & 1.0000 & 1.0000
& 100 & 1& 1.0000& 1.0000 & 1.0000 \\

100 & 1.0&1.0000 & 1.0000 & 0.9995
& 100 & 2&1.0000 & 1.0000 & 1.0000 \\

100 & 2.0&0.9955 & 0.9285 & 0.9330
& 100 & 4&0.9955 & 0.9170 & 0.9265 \\

100 & 4.0&0.6560 & 0.4995 & 0.5510
& 100 & 8&0.6560 & 0.4815 & 0.5475 \\

\hline
\end{tabular}
\label{table 6}
\end{table}

\begin{table}[]
\centering
\caption{Estimated power of $\gamma$, $\lambda_1$ and LR test for $q=0.4$}
\renewcommand{\arraystretch}{1.2}
\begin{tabular}{cccccccccc}
\hline
 &  & \multicolumn{3}{c}{\(H_0:r=1\)} &  &  & \multicolumn{3}{c}{\(H_0:r=2\)} \\
\cline{3-5} \cline{8-10}

\(n\) & \(\eta\) & \(\gamma\) & \(\lambda_1\) & LR
&
\(n\) & \(\eta\) & \(\gamma\) & \(\lambda_1\) & LR \\

\hline

20 & 0.5 &1.0000  & 0.9195 & 0.9230
& 20 & 1 &1.0000  & 0.9250 & 0.9330 \\

20 & 1.0 &0.9600  & 0.5360 & 0.6245
& 20 & 2 & 0.971 & 0.5765 & 0.6245 \\

20 & 2.0 &0.5370  & 0.2325 & 0.2915
& 20 & 4 &0.569  & 0.2355 & 0.3030 \\

20 & 4.0 &0.2245  & 0.1305 & 0.1585
& 20 & 8 &0.229  & 0.1385 & 0.1450 \\

\hline

40 & 0.5 &1.0000  & 1.0000 & 0.9995
& 40 & 1 &1.0000  & 0.9985 & 0.9985 \\

40 & 1.0 & 1.0000 & 0.8745 & 0.9140
& 40 & 2 & 1.0000 & 0.8895 & 0.9115 \\

40 & 2.0 &0.8345  & 0.4515 & 0.5325
& 40 & 4 & 0.8175 & 0.4450 & 0.5315 \\

40 & 4.0 &0.3850  & 0.1805 & 0.2445
& 40 & 8 & 0.3630 & 0.1900 & 0.2470 \\

\hline

60 & 0.5 &1.0000 & 1.0000 & 1.0000
& 60 & 1 & 1.0000 & 1.0000 & 1.0000 \\

60 & 1.0 &1.0000  & 0.9710 & 0.9840
& 60 & 2 & 1.0000 & 0.9695 & 0.9730 \\

60 & 2.0 &0.9420  & 0.6230 & 0.6920
& 60 & 4 & 0.9505 & 0.6215 & 0.6865 \\

60 & 4.0 &0.4915  & 0.2615 & 0.3095
& 60 & 8 &0.4840  & 0.2580 & 0.3105 \\

\hline

100 & 0.5 &1.0000  & 1.0000 & 1.0000
& 100 & 1 & 1.0000 & 1.0000 & 1.0000 \\

100 & 1.0 &1.0000  & 0.9995 & 0.9990
& 100 & 2 & 1.0000 & 0.9995 & 1.0000 \\

100 & 2.0 & 0.9965 & 0.8280 & 0.8735
& 100 & 4 & 0.9980 & 0.8265 & 0.8825 \\

100 & 4.0 &0.6835  & 0.4260 & 0.4695
& 100 & 8 & 0.6715 & 0.3980 & 0.4560 \\

\hline
\end{tabular}
\label{table 7}
\end{table}

\section{Real data analysis}
In this section, we consider two real datasets to illustrate the testing procedure of $\gamma$. \\
\textbf{Dataset 1:} We consider the data set presented in Table 6 of \cite{economou2009fitting}, which consists of a length-biased sample of failure times of communication units from aeroplanes. The sample was constructed by the authors from a larger population of confirmed failures using a length-biased sampling scheme, resulting in 69 observations, including 18 right-censored observations at 630 hours. In their study, the authors showed that the Weibull distribution provides an adequate fit to the data and fitted a length-biased Weibull model and obtained estimates $\hat{\beta}= 319.3$ and $\hat\eta= 1.118$.  \par \noindent
We use this dataset to assess whether our proposed test correctly detects the presence of length bias. Applying our test, the observed value of the test statistic is $\gamma=0.935$. By adopting the hypothesis of \eqref{hypothesis2}, the null hypothesis $H_0: r=0$ is rejected since $\gamma$ falls outside the $95\%$ acceptance region $\gamma= 0.935\notin (0.738,0.868)$. This confirms that our test successfully identifies the length-biased nature of the sample.
\\

\textbf{Dataset 2:}
 The data set consists of 50 observations representing the failure times (in weeks) of electronic items, originally reported by \cite{murthy2004weibull}. Before applying the proposed methodology, the suitability of the Weibull distribution for the data was assessed using the Kolmogorov--Smirnov (KS) test and the Anderson--Darling (AD) test. The corresponding p-values were 0.11 and 0.32, respectively, indicating that the Weibull distribution provides an adequate fit to the data. Subsequently, we fitted weibull distribution and obtained estimates of the shape parameter $\hat{\eta}= 0.8$ and the scale parameter $\hat\beta=6.97$, respectively.

To investigate whether the observed data exhibit length bias, the hypothesis \begin{equation*}
         H_0 : r = 0~~~vs.~~~ H_1: r = 1
\end{equation*}

was considered. The proposed test statistic was computed as $\gamma=0.424$. At the 5\% significance level, the acceptance region for $H_0$ is $(0.223,0.532)$. Since $\gamma=0.424 \in (0.223,0.532)$, the null hypothesis is not rejected. Therefore, there is sufficient statistical evidence to conclude that the data is not length biased.

\section{Conclusion}
In this paper, a hypothesis testing procedure based on the Rényi divergence measure has been proposed for detecting 
biased sampling, with particular emphasis on length-biased and size-biased sampling. We derived the empirical Rényi divergence statistic $\gamma$ and established its theoretical properties, including scale invariance and asymptotic normality. 
Through an extensive simulation study, critical values of the proposed test were obtained through Monte Carlo simulation for Weibull distributions under both uncensored and censored framework.

A comprehensive power comparison study against the Kullback-Leibler divergence-based test and the likelihood ratio test of \cite{economou2014kullback} demonstrated that the proposed test performs competitively across most settings and shows particular strength as the sample size and censoring proportion increase. The applicability of the proposed methodology was demonstrated on two real datasets. In the first, a length-biased sample of aircraft communication unit failure times, the test correctly rejected the null hypothesis of no bias, confirming its ability to detect a length-biased sampling scheme known to be present by construction. In the second, a dataset of electronic component failure times, the test failed to reject the null hypothesis, correctly indicating the absence of length bias in that sample. These results support the practical reliability of the proposed test as a diagnostic tool prior to model fitting in survival and reliability analysis.

The proposed approach is computationally simple, does not depend on the scale parameter of the underlying distribution, and can be readily implemented in practice. Although the present work focuses on the Weibull distribution, the methodology can be extended to other parametric families satisfying the required regularity conditions. Future work may consider extensions to progressive censoring schemes, progressively weighted sampling mechanisms, and semiparametric or nonparametric settings, thereby broadening the applicability of the proposed Rényi divergence-based framework.



	
\section*{Conflict of interest statement}\noindent	
	On behalf of all authors, the corresponding author states that there is no conflict of interest.
	
 \bibliographystyle{apalike}
\bibliography{testing}	

@article{fisher1934effect,
  title={The effect of methods of ascertainment upon the estimation of frequencies},
  author={Fisher, Ronald A},
  journal={Annals of Eugenics},
  volume={6},
  number={1},
  pages={13--25},
  year={1934},
  publisher={Wiley Online Library}
}

@article{rao1965discrete,
  title={On discrete distributions arising out of methods of ascertainment},
  author={Rao, C Radhakrishna},
  journal={Sankhy{\=a}: The Indian Journal of Statistics, Series A},
  pages={311--324},
  year={1965},
  publisher={JSTOR}
}

@article{patil1976size,
  title={On size-biased sampling and related form-invariant weighted distributions},
  author={Patil, Ganapati P and Ord, JK},
  journal={Sankhy{\=a}: The Indian Journal of Statistics, Series B},
  pages={48--61},
  year={1976},
  publisher={JSTOR}
}

@article{patil1984studies,
  title={Studies in statistical ecology involving weighted distributions},
  author={Patil, GP},
  journal={Statistics: Applications and New Directions},
  pages={478--503},
  year={1984},
  publisher={Indian Statistical Institute Calcutta}
}

@inproceedings{renyi1961measures,
  title={On measures of entropy and information},
  author={R{\'e}nyi, Alfr{\'e}d},
  booktitle={Proceedings of the fourth Berkeley symposium on mathematical statistics and probability, volume 1: contributions to the theory of statistics},
  volume={4},
  pages={547--562},
  year={1961},
  organization={University of California Press}
}

@article{kullback1951information,
  title={On information and sufficiency},
  author={Kullback, Solomon and Leibler, Richard A},
  journal={The annals of mathematical statistics},
  volume={22},
  number={1},
  pages={79--86},
  year={1951},
  publisher={JSTOR}
}

@article{shannon1948mathematical,
  title={A mathematical theory of communication},
  author={Shannon, Claude E},
  journal={The Bell system technical journal},
  volume={27},
  number={3},
  pages={379--423},
  year={1948},
  publisher={Nokia Bell Labs}
}

@article{sadeghpour2018exponentiality,
  title={Exponentiality test based on Renyi distance between equilibrium distributions},
  author={Sadeghpour, Mehran and Baratpour, S and Habibirad, Arezou},
  journal={Communications in Statistics-Simulation and Computation},
  volume={47},
  number={10},
  pages={2960--2967},
  year={2018},
  publisher={Taylor \& Francis}
}

@article{economou2014kullback,
  title={Kullback--Leibler divergence measure based tests concerning the biasness in a sample},
  author={Economou, Polychronis and Tzavelas, George},
  journal={Statistical Methodology},
  volume={21},
  pages={88--108},
  year={2014},
  publisher={Elsevier}
}

@book{murthy2004weibull,
  title={Weibull models},
  author={Murthy, DN Prabhakar and Xie, Min and Jiang, Renyan},
  year={2004},
  publisher={John Wiley \& Sons}
}

@article{economou2009fitting,
  title={Fitting parametric frailty and mixture models under biased sampling},
  author={Economou, P and Caroni, C},
  journal={Journal of Applied Statistics},
  volume={36},
  number={1},
  pages={53--66},
  year={2009},
  publisher={Taylor \& Francis}
}

@misc{patil2002weighted,
  title={Weighted distributions,Encyclopedia of Environmetrics, Vol. 4},
  author={Patil, GP},
  year={2002},
  publisher={John Wiley \& Sons, Chichester}
}

@article{navarro2003detect,
  title={How to detect biased samples?},
  author={Navarro, J and Ruiz, JM and Del Aguila, Y},
  journal={Biometrical Journal: Journal of Mathematical Methods in Biosciences},
  volume={45},
  number={1},
  pages={91--112},
  year={2003},
  publisher={Wiley Online Library}
}

@article{akman2007simple,
  title={A simple test for detection of length-biased sampling},
  author={Akman, OLCAY and Gamage, Jinadasa and Jannot, Jason and Juliano, STEVEN and Thurman, Andrew and Whitman, Douglas},
  journal={JP Journal of Biostatistics},
  volume={1},
  number={2},
  pages={189--195},
  year={2007}
}
\end{document}